\documentclass{article}
\usepackage{amsmath}
\usepackage{latexsym}
\usepackage{graphicx}
\usepackage{color}
\usepackage{amssymb}
\usepackage{subfigure}
\usepackage{tabularx}
\usepackage[labelfont={bf,sf},tableposition=top]{caption}

\newtheorem{theorem}{Theorem}
\newtheorem{acknowledgement}{Acknowledgement}
\newtheorem{definition}{Definition}

\newtheorem{proposition}{Proposition}
\newtheorem{remark}{Remark}
\newtheorem{conjecture}{Conjecture}
\newtheorem{proof}{Proof}

\begin{document}
\title{Study of the forward Dirichlet boundary value problem for the two-dimensional Electrical Impedance Equation}
\author{M. P. Ramirez T.\\
\small{Communications and Digital Signal Processing Group,}\\ 
\small{Faculty of Engineering, La Salle University,}\\
\small{B. Franklin 47, C.P. 06140, Mexico.}}

\date{}
\maketitle
\begin{abstract}
Using a conjecture that allows to approach separable-variables conductivity functions, the elements of the Modern Pseudoanalytic Function Theory are used, for the first time, to numerically solve the Dirichlet boundary value problem of the two-dimensional Electrical Impedance Equation, when the conductivity function arises from geometrical figures, located within bounded domains.
\end{abstract}

\section{Introduction}

The study of the solutions corresponding to the Electrical Impedance Equation
\begin{equation}
\nabla\cdot\left(\sigma\nabla u\right)=0,
\label{int:00}
\end{equation}
where $\sigma$ represents the conductivity and $u$ denotes the electric potential, is fundamental for the proper understanding of a wide variety of boundary value problems, that posses special relevance in different branches of Mathematical Physics. Among many important problems, it is particularly interesting the one corresponding to the Electrical Impedance Tomography, because of its importance in Applied Physics and Engineering. As a matter of fact, only taking into account the Medical Imaging applications, no doubts remain about its relevance.

The study of this inverse problem, correctly posed in mathematical terms by A.P. Calderon \cite{calderon}, is based upon iterative methods that employ solutions of the forward problem, attempting to fulfil a certain boundary condition (usually upcoming from physiological measurements), by introducing variations in the conductivity function at every step, in order to reduce the difference between the theoretical $u$ and the measured one (see e.g., the classical work \cite{webster}).

Yet, the mathematical complexity of (\ref{int:00}) posed so strong challenges, that the Electrical Impedance Tomography remained only as an alternative Medical Imaging technique, seldom considered among the basic clinical applications. As a matter of fact, the Electrical Impedance Tomography is still considered an ill posed problem.

In this direction, the mathematical foundations dedicated to this topic, perhaps received one of the most important contributions when V. Kravchenko in 2005 \cite{kra2005}, and independently K. Astala and L. P\"aiv\"arinta in 2006 \cite{astala}, discovered that the two-dimensional case of (\ref{int:00}) was completely equivalent to a special kind of Vekua equation \cite{vekua}, which had been deeply studied in a variety of works, been the most important those published by L. Bers \cite{bers} and by I. Vekua \cite{vekua}.

The list of novel works arising after such discovering, is as long as it is interesting. Still, it shall be pointed out that just some of them are dedicated to the Engineering applications, since it is not clear how to model a wide variety of physical phenomena, in order to make them susceptible to be analysed by means the Modern Theory of Pseudoanalytic Functions (see e.g. \cite{kpa}).

This work intends to make a positive contribution in this direction. Already in \cite{ioprrh} was posed a basic idea for interpolating values of conductivity within bounded domains, in order to obtain a certain class of analytic representations: Two-dimensional separable-variables functions, one of the keys to fully applied the novel mathematical methods into Engineering problems.

In this work, a more general methodology is posed, based upon the main Conjecture arising from \cite{ioprrh}. Starting with some examples for which the conductivity functions are known in exact form, the elements of the Pseudoanalytic Function Theory are used, for the first time, to approach solutions of the forward boundary value problem for (\ref{int:00}), considering conductivity distributions upcoming from geometrical figures.

Even the full set of examples is provided for a circular domain, the results can be extended to a wide variety of bounded domains. Thus, the methods provided along these pages, could well be ready for studying images corresponding already to physical experimental models.

\section{Preliminaries}

Following \cite{bers}, let the complex-valued functions $F,G$ satisfy the condition
\begin{equation}
\mbox{Im}\left(\overline{F}G\right)>0,
\label{pre:00}
\end{equation}
where $\overline{F}$ denotes the complex conjugation of $F$: $\overline{F}=\mbox{Re}F-i\mbox{Im}F$, and $i$ is the standard imaginary unit $i^{2}=-1$. Thus, any complex-valued function $W$ can be expressed by means of the linear combination of $F$ and $G$:
\[
W=\phi F+\psi G,
\]
where $\phi$ and $\psi$ are purely real functions. Two complex functions that fulfil (\ref{pre:00}) shall be called a \emph{generating pair} $(F,G)$. Bers \cite{bers} introduced the $(F,G)$-derivative of the function $W$ according to the expression:
\begin{equation}
\partial_{(F,G)}W=\left(\partial_{z}\phi\right)F+\left(\partial_{z}\psi\right)G.
\label{pre:01}
\end{equation}
But this derivative will exist if and only if
\begin{equation}
\left(\partial_{\overline{z}}\phi\right)F+\left(\partial_{\overline{z}}\psi\right)G=0.
\label{pre:02}
\end{equation}
Here
\[
\partial_{z}=\partial_{x}-i\partial_{y},\ \ \ \partial_{\overline{z}}=\partial_{x}+i\partial_{y}.
\]
Notice that these operators are classically introduced with the factor $\frac{1}{2}$, but it will result somehow more convenient to omit it int his work.

Introducing the functions
\begin{eqnarray}
A_{(F,G)}=\frac{\overline{F}\partial_{z}G-\overline{G}\partial_{z}F}{F\overline{G}-G\overline{F}},\ \ \ a_{(F,G)}=-\frac{\overline{F}\partial_{\overline{z}}G-\overline{G}\partial_{\overline{z}}F}{F\overline{G}-G\overline{F}}, \nonumber \\
B_{(F,G)}=\frac{F\partial_{z}G-G\partial_{z}F}{F\overline{G}-G\overline{F}},\ \ \ b_{(F,G)}=-\frac{G\partial_{\overline{z}}F-F\partial_{\overline{z}}G}{F\overline{G}-G\overline{F}};
\label{pre:03}
\end{eqnarray}
the expression (\ref{pre:01}) of the $(F,G)$-derivative will turn into
\begin{equation}
\partial_{(F,G)}W=\partial_{z}W-A_{(F,G)}W-B_{(F,G)}\overline{W},
\label{pre:04}
\end{equation}
as well the condition (\ref{pre:02}) can be written as
\begin{equation}
\partial_{\overline{z}}W-a_{(F,G)}W-b_{(F,G)}\overline{W}=0.
\label{pre:05}
\end{equation}

The functions defined in (\ref{pre:03}) are called the \emph{characteristic coefficients} of the generating pair $(F,G)$, and the functions $W$ satisfying the condition (\ref{pre:05}) are named $(F,G)$-pseudoanalytic. Indeed, the equation(\ref{pre:05}) is know as the \emph{Vekua equation} \cite{vekua}, and in many senses is the foundation of the present work.

The following sentences were originally presented in \cite{bers} and \cite{kpa}, and they have been slightly adapted for the purposes of this paper.
\begin{theorem}
\label{th:00}
The elements of the generating pair $(F,G)$ are $(F,G)$-pseudoanalytic:
\[
\partial_{(F,G)}F=\partial_{(F,G)}G=0.
\]
\end{theorem}
\begin{remark}
Let $p$ be a non-vanishing function within a bounded domain $\Omega\in\mathbb{R}^{2}$. The functions
\begin{equation}
F_{0}=p, \ \ \ G_{0}=\frac{i}{p},
\label{pre:06}
\end{equation}
constitute a generating pair, whose characteristic coefficients are
\begin{eqnarray}
A_{\left(F_{0},G_{0}\right)}=a_{\left(F_{0},G_{0}\right)}=0,\nonumber \\
B_{\left(F_{0},G_{0}\right)}=\frac{\partial_{z}p}{p},\ \ \ a_{\left(F_{0},G_{0}\right)}=\frac{\partial_{\overline{z}}p}{p}.
\label{pre:07}
\end{eqnarray} 
\end{remark}
\begin{definition}
Let $\left(F_{0},G_{0}\right)$ and $\left(F_{1},G{1}\right)$ be two generating pairs of the form (\ref{pre:07}), and let their characteristic coefficients satisfy the relation
\begin{equation}
B_{\left(F_{1},G_{1}\right)}=-b_{\left(F_{0},G_{0}\right)}.
\label{pre:08}
\end{equation} 
Thus the pair $\left(F_{1},G_{1}\right)$ will be called a successor of the pair $\left(F_{0},G_{0}\right)$, whereas $\left(F_{0},G_{0}\right)$ will be called a predecessor of $\left(F_{1},G_{1}\right)$.
\end{definition}
\begin{definition}
Let 
\[
\left\lbrace \left(F_{m},G_{m}\right) \right\rbrace,\ m=0,\pm 1,\pm 2,...
\]
be a set of generating pairs, and let every $\left(F_{m+1},G_{m+1}\right)$ be a successor of $\left(F_{m},G_{m}\right)$. Thus, the set $\left\lbrace \left(F_{m},G_{m}\right) \right\rbrace$ is called a generating sequence. Moreover, if there exist a number $k$ such that $\left(F_{m},G_{m}\right)=\left(F_{m+k},G_{m+k}\right)$ the generating sequence is said to be periodic, with period $k$.

Finally, if $\left(F,G\right)=\left(F_{0},G_{0}\right)$, the generating pair $\left(F,G\right)$ will be embedded into the generating sequence $\left\lbrace \left(F_{m},G_{m}\right) \right\rbrace$.
\end{definition}

\begin{theorem}
\label{th:01}
Let $\left(F,G\right)$ be a generating pair of the form (\ref{pre:06}), and let $p$ be a separable-variables function:
\[
p=p_{1}(x)p_{2}(y),
\]
where $x,y\in\mathbb{R}$. Thus $\left(F,G\right)$ is embedded into a periodic generating sequence, with period $k=2$, such that
\[
F_{m}=\frac{p_{2}(y)}{p_{1}(x)},\ \ G_{m}=i\frac{p_{1}(x)}{p_{2}(y)};
\]
when $m$ is even, and  
\[
F_{m}=p_{1}(x)p_{2}(y),\ \ G_{m}=\frac{i}{p_{1}(x)p_{2}(y)};
\]
when $m$ is odd. 

Moreover, if in particular $p_{1}(x)=1$, it is easy to see that the generating sequence in which $\left(F,G\right)$ is embedded will be also periodic, but it will posses period $k=1$.
\end{theorem}

L. Bers also introduced the concept of the $\left(F_{0},G_{0}\right)$-integral of a complex-valued function $W$. The detailed conditions for its existence can be found in \cite{bers}.

\begin{definition}
Let $\left(F_{0},G_{0}\right)$ be a generating pair of the form. Its adjoin generating pair $\left(F_{0}^{*},G_{0}^{*}\right)$ is defined according to the formulas
\[
F_{0}^{*}=-iF_{0},\ \ G_{0}^{*}=-iG_{0}.
\]
\end{definition}
\begin{definition}
The $\left(F_{0},G_{0}\right)$-integral of a complex-valued function $W$ (when it exists \cite{bers}) is defined as:
\[
\int_{\Gamma}Wd_{\left(F_{0},G_{0}\right)}z=F_{0}\mbox{Re}\int_{\Gamma}G_{0}^{*}Wdz+G_{0}\mbox{Re}\int_{\Gamma}F_{0}^{*}Wdz,
\]
where $\Gamma$ is a rectifiable curve within a domain $\Omega\in\mathbb{C}$. Specifically, if we consider the $\left(F_{0},G_{0}\right)$-integral of the $\left(F_{0},G_{0}\right)$-derivative of $W$, we  will have:
\begin{equation}
\int_{z_{0}}^{z}\partial_{\left(F_{0},G_{0}\right)}Wd_{\left(F_{0},G_{0}\right)}z=-\phi (z_{0})F(z)-\psi (z_{0})G(z)+W(z).
\label{pre:09}
\end{equation}
Here $z=x+iy$, and $z_{0}$ is a fixed point in the complex plane. But according to the Theorem \ref{th:00}, the $\left(F_{0},G_{0}\right)$-derivative of $F$ and of $G$ vanish identically, thus the expression (\ref{pre:09}) can be considered as the $\left(F_{0},G_{0}\right)$-antiderivative of $\partial_{\left(F_{0},G_{0}\right)}W$.
\end{definition}

\subsection{Formal Powers}

\begin{definition}
\label{def:00}
The formal power $Z_{m}^{(0)}\left(a_{0},z;z_{0}\right)$ belonging to the generating pair $\left(F_{m},G_{m}\right)$, with formal degree $(0)$, complex constant coefficient $a_{0}$, center at $z_{0}$, and depending upon $z=x+iy$, is defined according to the expression:
\[
Z_{m}^{(0)}\left(a_{n},z;z_{0}\right)=\lambda F(z)+\mu G(z),
\]
where $\lambda$ and $\mu$ are complex constants that fulfil the condition:
\[
\lambda F(z_{0})+\mu G(z_{0})=a_{0}.
\]
The formal powers with higher degrees are calculated according to the recursive formulas:
\begin{equation}
Z_{m}^{(n)}\left(a_{n},z;z_{0}\right)=n\int_{z_{0}}^{z}Z_{m-1}^{(n-1)}\left(a_{n},z;z_{0}\right)d_{\left(F_{m},G_{m}\right)}z,
\label{pre:10}
\end{equation}
where $n=1,2,3,...$ Notice the integral operators in the right side of the last expression are $\left(F_{m},G_{m}\right)$-antiderivatives.
\end{definition}

\begin{theorem}
The formal powers posses the following properties:
\begin{enumerate}
\item Every $Z_{m}^{(n)}\left(a_{n},z;z_{0}\right),\ n=0,1,2,...$ is an $\left(F_{m},G_{m}\right)$-pseudoanalytic function.
\item Let $a_{n}=a'_{n}+ia''_{n}$, where $a'_{n}$ and $a''_{n}$ are real constants. The following relation holds
\begin{equation}
Z_{m}^{(n)}\left(a_{n},z;z_{0}\right)=a'_{n}Z_{m}^{(n)}\left(1,z;z_{0}\right)+a''Z_{m}^{(n)}\left(i,z;z_{0}\right).
\end{equation}
\end{enumerate}
\end{theorem}
\begin{theorem}
Let $W$ be an $\left(F_{m},G_{m}\right)$-pseudoanalytic function. Then it can be expressed by means of the so-called Taylor series in formal powers:
\begin{equation}
W=\sum_{n=0}^{\infty}Z_{m}^{(n)}\left(a_{n},z;z_{0}\right).
\label{pre:11}
\end{equation}
Since any $\left(F_{m},G_{m}\right)$-pseudoanalytic function $W$ accepts this expansion, this is an analytical representation for the general solution of the Vekua equation (\ref{pre:07}).  
\end{theorem}
\section{The two-dimensional Electrical Impedance Equation}
Let us consider the two-dimensional case of the equation (\ref{int:00}): 
\[
\nabla\cdot\left(\sigma\nabla u\right)=0.
\]
As it has been shown in several previous works (see e.g. \cite{kpa} and \cite{oct}), if $\sigma$ can be expressed by means of a separable-variables function
\[
\sigma(x,y)=\sigma_{1}(x)\sigma_{2}(y),
\]
introducing the notations
\begin{eqnarray}
W=\sqrt{\sigma}\partial_{x}u-i\sqrt{\sigma}\partial_{y}u, \nonumber \\
p=\sqrt{\frac{\sigma_{2}}{\sigma_{1}}};
\label{eie:00}
\end{eqnarray}
the equation (\ref{int:00}) will turn into the Vekua equation
\begin{equation}
\partial_{\overline{z}}W-\frac{\partial_{\overline{z}}p}{p}\overline{W}=0,
\label{eie:01}
\end{equation}
for which the functions
\begin{equation}
F_{0}=p,\ \ G_{0}=\frac{i}{p},
\label{eie:02}
\end{equation}
conform a generating pair. From (\ref{eie:00}) it is easy to see that this pair is embedded into a generating sequence with period $k=2$, because $p$ is separable-variables, according to Theorem \ref{th:01}.

\subsection{A complete set for boundary value problems of the Electrical Impedance Equation}

The possession of an explicit generating sequence, allows the construction of the formal powers (\ref{pre:10}), so we can approach any solution for (\ref{eie:01}), which will be closely related with the solutions of (\ref{int:00}) according to the relations (\ref{eie:00}).

As a matter of fact, one special and very important relation between the solutions of (\ref{int:00}) and of (\ref{eie:01}) was elegantly posed in \cite{cck}, and this idea will play a central role in the present work.

\begin{theorem}
\cite{cck} Let us consider the set of formal powers
\[
\left\lbrace Z_{0}^{(n)}\left(1,z;0\right),\ Z_{0}^{(n)}\left(i,z;0\right) \right\rbrace_{n=0}^{\infty},
\]
corresponding to the generating pair (\ref{eie:02}), and let $\Omega\in\mathbb{R}^2$ be a bounded domain such that $0\in\Omega$, but $0\notin\Gamma$. Then the set of functions valued on $\Gamma$:
\begin{equation}
\left\lbrace \mbox{Re}Z_{0}^{(n)}\left(1,z;0\right)\vert_{\Gamma},\ \mbox{Re}Z_{0}^{(n)}\left(i,z;0\right)\vert_{\Gamma} \right\rbrace_{n=0}^{\infty},
\label{eie:03}
\end{equation}
conform a complete system for approaching solutions of the Dirichlet boundary value problem corresponding to (\ref{int:00}).
\end{theorem}

The last statement implies that any boundary condition $u\vert_{\Gamma}$, imposed for the solutions $u$ of (\ref{int:00}), can be approached according to the expression:
\[
\lim_{N\rightarrow\infty}u\vert_{\Gamma}-\sum_{n=0}^{N}\left(a'\mbox{Re}Z_{0}^{(n)}\left(1,z;0\right)\vert_{\Gamma}+a''\mbox{Re} Z_{0}^{(n)}\left(i,z;0\right)\vert_{\Gamma}\right)=0,
\] 
where $a'$ and $a''$ are real constants.

Summarizing, when a separable-variables conductivity function $\sigma$ is given within a bounded domain $\Omega$, and a boundary condition $u\vert_{\Gamma}$ is imposed for the solution of (\ref{int:00}), it will be always possible to construct a finite set of functions, subset of (\ref{eie:03}), such that
\begin{equation}
\oint \left[ u\vert_{\Gamma}-\left(\sum_{n=0}^{N}a'\mbox{Re}Z_{0}^{(n)}\left(1,z;0\right)\vert_{\Gamma}+a''\ \mbox{Re}Z_{0}^{(n)}\left(i,z;0\right)\vert_{\Gamma}\right)\right]^{2}dl<\epsilon,
\label{eie:04}
\end{equation}
where $\epsilon>0$ and $l\in\Gamma$.

\subsection{Construction of a piecewise separable-variables function}

One of the main objectives of this work is to contribute into the construction of a new theory for the Electrical Impedance Tomography problem. Hence it is natural to search for the mathematical tools that will allow us to fully apply the modern Pseudoanalytic Function Theory into the analysis, e.g., of medical images.

This means that it is necessary to introduce interpolation methods that, given a set of conductivity values defined into a bounded domain on the plane, 
can reach separable-variables functions. One of the first approaches in this direction was posed in \cite{oct}, and it was properly analysed in \cite{ioprrh}. Indeed, the last reference is completely dedicated to prove the following assessment.

\begin{conjecture}
\label{con:00}
\cite{ioprrh} Let $\sigma$ be a function defined within a bounded domain $\Omega \in\mathbb{R}^{2}$, possessing discontinuities only of the first kind. Then it is possible to approach $\sigma$ by means of a piecewise separable-variables function of the form:
\begin{displaymath}
   \sigma_{pw} (x,y) = \left\{
     \begin{array}{lr}
       \frac{x+K_{1}}{\chi_{1}+K_{1}}\cdot f_{1}(y) & : x \in [x_{(0)},x_{(1)});\\
       \frac{x+K_{2}}{\chi_{2}+K_{2}}\cdot f_{2}(y) & : x \in [x_{(1)},x_{(2)});\\
       \cdots & \\
       \frac{x+K_{M}}{\chi_{M}+K_{M}}\cdot f_{M}(y) & : x \in [x_{(M-1)},x_{(M)}].
     \end{array}
   \right.   
\end{displaymath}
\begin{equation}
\label{eie:05}
\end{equation}
This separable-variables function can be employed for numerically constructing a finite set of formal powers of the form (\ref{eie:03}), in order to approach solutions for the Dirichlet boundary value problem of the two-dimensional Electrical Impedance Equation (\ref{int:00}), in the sense of (\ref{eie:04}), when a boundary condition $u\vert_{\Gamma}$ is imposed.
\end{conjecture}

A simplified illustration of the procedure for the construction of such piecewise function, is provided along in Figure 1, and a brief explanation of the central ideas will be provided now.
\begin{figure}
\centering
\subfigure[A circular domain sectioned.]{
\includegraphics[scale=0.350]{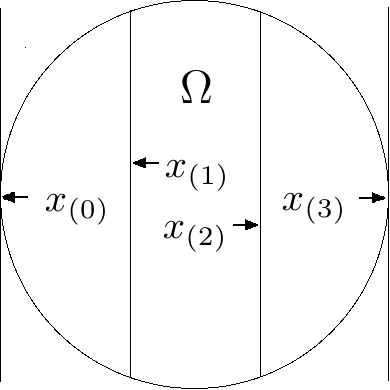}
\label{fig:eie00a}
}
\subfigure[Interpolating functions $f(y)$.]{
\includegraphics[scale=0.350]{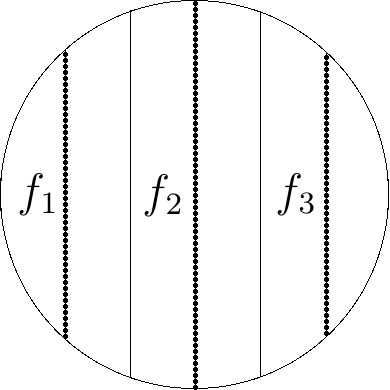}
\label{fig:eie00b}
}

\caption{Simplified illustrations of the steps for approaching piecewise separable-variables conductivity functions.}
\end{figure}

For simplicity, we will consider hereafter the domain $\Omega$ as the unitary circle, noticing that all postulates can be extended to a wide set of bounded domains, that are of special interest in many branches of Applied Mathematics, Physics, and Engineering.

The first step is to divide the domain $\Omega$ into a number $M$ of subsections, as shown in the Figure \ref{fig:eie00a}. Along a strait line crossing every subsection, parallel to the $y$-axis, we will collect a number $q$ of conductivity values, in order to introduce an interpolating function $f$ depending only upon the spatial variable $y$. We will posses then a set of $M$ interpolating functions, one for every subsection, as shown in Figure \ref{fig:eie00b}.

Finally, we state that the conductivity function in every subregion will have the separable-variables form
\[
\sigma=\frac{x+K_{j}}{\chi_{j}+K_{j}}\cdot f_{j}(y),
\]
where $j$ is the number of the subregion, $\chi_{j}$ is the common $x$-coordinate of all points collected in the $j$-subregion, and $K_{j}$ is a positive real constant such that $x+K_{j}\neq 0,\ x\in\Omega$.

The postulate remains a conjecture because not any formal extension of the theorem posed in \cite{cck}, about the completeness of the set (\ref{eie:03}), is known for the case of piecewise separable-variables functions within bounded domains.

Yet, the full set of examples presented in \cite{ioprrh} shows that the representation (\ref{eie:05}) is useful for solving the Dirichlet boundary value problem of (\ref{int:00}), since the numerical calculations succeed to approach the boundary conditions provided by known exact solutions, for a variety of conductivities that were both separable-variables and non separable-variables by definition.

This property will be now extended for analysing conductivities upcoming from geometrical distributions, precisely as those employed in physical measurements \cite{webster}. To achieve this, we will consider a limit case of the Conjecture \ref{con:00} that will result specially useful for these classes of conductivities. In behalf of simplicity, $\Omega$ will be considered as the unitary circle, but it is possible to verify that the next statements can be generalized for a wide variety of bounded domains.

\begin{proposition}
\label{pro:00}
Every conductivity function $\sigma$, defined within a bounded domain $\Omega\in\mathbb{R}^{2}$, and possessing only discontinuities of the first kind, can be considered the limit case of a piecewise separable-variables conductivity function $\sigma_{pw}$ of the form (\ref{eie:05}), when the number $M$ of subsections, and the number $q$ of collected values at every subsection, tend to infinity. This is
\begin{equation}
\sigma(x,y)=\lim_{M,q\rightarrow\infty}\sigma_{pw}(x,y).
\label{eie:06}
\end{equation}
Moreover, since
\[
\lim_{M\rightarrow\infty}\frac{x+K_{j}}{\chi_{j}+K_{j}}=1,
\]
where $j$ is the number of the subsection, it follows from Theorem \ref{th:01}, that the generating sequence corresponding to this limit case, employed for numerically approaching the formal powers, will be periodic with period $k=1$. 
\end{proposition}

\begin{proof}
Let us consider that $\Omega$ has been sectioned into $M$ subdomains $\left\lbrace\Omega_{j}\right\rbrace_{j=1}^{M}$, by employing the set of equidistant $y$-axis parallel lines
\[
\left\lbrace \xi_{j}=\ x_{(j)}:x_{(j)}=-1+\frac{2j}{M} ;\ j=0,1,...,M\right\rbrace,
\]
Since $\Omega$ coincides with the unitary circle, it can be described by the union of $M$ subdomains, defined according to the expressions
\begin{eqnarray}
\Omega_{j+1}&=&\left\lbrace (x,y)\ : x\in [x_{(j)},x_{(j+1)}),\ x^{2}+y^{2}\leq 1\right\rbrace,\ j=0,1,2,...,M-2;\nonumber \\
\Omega_{M}&=&\left\lbrace (x,y)\ : x\in [x_{(M-1)},x_{(M)}],\ x^{2}+y^{2}\leq 1\right\rbrace. \nonumber
\end{eqnarray}
Let us consider also the set of lines
\[
\left\lbrace \varphi_{j+1}=\chi_{j+1}:\chi_{j+1}=x_{(j)}+\frac{x_{(j)}+x_{(j+1)}}{2};\ j=0,1,...,M-1\right\rbrace.
\]
It is clear that
\[
\lim_{M\rightarrow\infty}\vert x_{(j)}-x_{(j+1)}\vert=0,\ j=0,1,...,M;
\]
and in consequence, for $\forall x\in\Omega_{j};\ j=1,2,...,M$; we will have that
\[
x\rightarrow\chi_{j}.
\]
It immediately follows that 
\begin{equation}
\lim_{M\rightarrow\infty}\frac{x+K_{j}}{\chi_{j}+K_{j}}=1.
\label{eie:07}
\end{equation}
This implies that every $\Omega_{j}$ will be confined into a segment of a parallel line to the $y$-axis. 

Indeed, when the $M\rightarrow\infty$, $\Omega$ can be represented by the set of the line segments
\begin{equation}
\left\lbrace \varphi=C: -1<C<1,\ x^{2}+y^{2}<1\right\rbrace.
\label{eie:08}
\end{equation}

Beside, since the number $q$ of conductivity values, collected along the line $\varphi_{j+1}$ within every $\Omega_{j}$, also tends to infinite, not any interpolation method will be required for obtaining the $f_{j}(y)$ functions. They will simply coincide with the original values of the function $\sigma$ over the elements of the set (\ref{eie:08}).

Finally, from (\ref{eie:07}), it follows that
\[
\lim_{M\rightarrow\infty}\frac{x+K_{j}}{\chi_{j}+K_{j}}\cdot f_{j}(y)=f_{j}(y),
\]
which, according to the last sentence of the Theorem \ref{th:01}, will provoke that the corresponding generating sequence, employed to numerically approach some of the formal powers (\ref{eie:03}), will be periodic with period $k=1$.
\end{proof}

\section{Numerical solutions for the two-dimensional Electrical Impedance Equation}
We now analyse a selected set of analytic conductivity functions for which exact solutions are known, in order to examine the effectiveness of the method posed in Proposition \ref{pro:00}. Thereafter, we will consider conductivity distributions upcoming from geometrical distributions, whose analytical representation is, in general, unknown, imposing certain boundary conditions that will help us to appreciate the behaviour of the technique in this special and important cases.

A detailed description of the numerical methods used to approach the solution, that fulfil the boundary condition of every example, can be found in \cite{bucio}. Since we are considering the unitary circle, and taking into account the validity of the expression (\ref{eie:04}), that as a matter of fact is a Lebesgue integral-type operator, let us introduce an inner product for the elements of the finite set

\begin{equation}
\left\lbrace \mbox{Re}Z_{0}^{(n)}\left(1,z;0\right)\vert_{\Gamma},\ \mbox{Re}Z_{0}^{(n)}\left(i,z;0\right)\vert_{\Gamma} \right\rbrace_{n=0}^{N},
\label{ns:00}
\end{equation}
according to the formula
\[
\left\langle \mbox{Re}Z_{0}^{(n_{1})}\vert_{\Gamma},\mbox{Re}Z_{0}^{(n_{2})}\vert_{\Gamma}\right\rangle= \oint \mbox{Re}Z_{0}^{(n_{1})}(l)\cdot\mbox{Re}Z_{0}^{(n_{2})}(l)dl,
\]
where $l\in\Gamma$, and $n_{1},n_{2}=0,1,2,...$ It follows that we can obtain a set of $2N+1$ orthonormal functions $\left\lbrace u_{\alpha}\right\rbrace_{\alpha=0}^{2N-1}$, such that we can approach an imposed boundary condition $u\vert_{\Gamma}$ according to the expression
\[
u\vert_{\Gamma}\sim\sum_{\alpha=0}^{2N+1}b_{\alpha}u_{\alpha},
\]
where $b_{\alpha}$ are real constant coefficients. Notice that the apparent lose of one base function is because, by virtue of the Definition \ref{def:00}, we have
\[
Z_{0}^{(0)}(i,0;z)=\frac{i}{p},
\]
thus $\mbox{Re}Z_{0}^{(0)}(i,0;z)=0$.
It should be also noticed that the orthonormalization procedure has been performed considering, first, the subset of $N+1$ functions
\begin{equation}
\left\lbrace \mbox{Re}Z_{0}^{(n)}(1,0;z)\right\rbrace_{n=0}^{N},
\end{equation}
\label{ns:04}
followed by the subset of $N$ functions
\begin{equation}
\left\lbrace \mbox{Re}Z_{0}^{(n)}(i,0;z)\right\rbrace_{n=1}^{N}.
\label{ns:05}
\end{equation}
This remark is important for adequately examining the illustrations where the absolute values of the coefficients $b_{\alpha}$ are displayed.

\subsection{The sinusoidal case.}
\begin{proposition}
Let
\begin{equation}
\sigma=(2+\cos{\omega x})(2+\sin{\omega y}).
\label{ns:01}
\end{equation}
Then the function
\begin{equation}
u=\frac{2}{\sqrt{3}}\arctan{\left(\frac{\tan{\frac{\omega x}{2}}}{\sqrt{3}}\right)}+\frac{2}{\sqrt{3}}\arctan\left(\frac{1+2\tan\frac{\omega y}{2}}{\sqrt{3}} \right);
\label{ns:02}
\end{equation}
is a particular solution of the Electrical Impedance Equation (\ref{int:00}).
\end{proposition}
This case was selected because of the variations of conductivity that take place within the unitary circle. Unfortunately, for avoiding the indetermination of the tangent functions contained into the particular solution (\ref{ns:02}), we could only consider the case $\omega\rightarrow\pi$. An illustration of this conductivity is given in the Figure 2.
\begin{figure}
\centering
\includegraphics[scale=0.30]{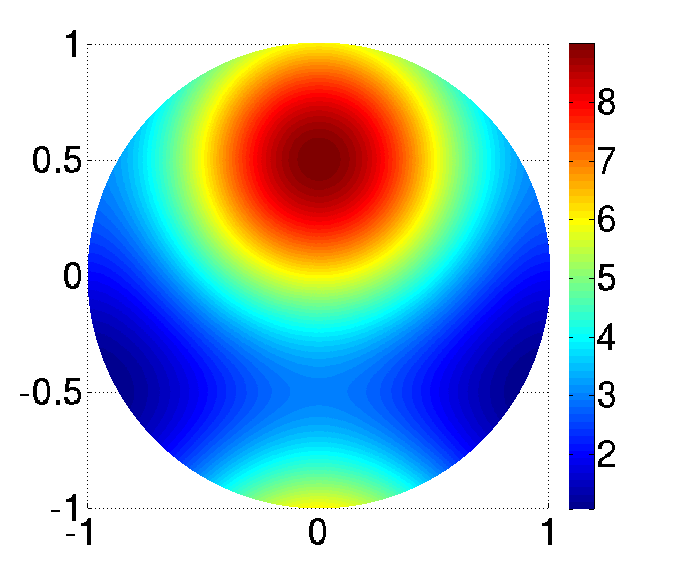}
\label{fig:ns00}
\caption{$\sigma=(2+\cos\pi x)(2+\sin\pi y)$ within the unitary circle.}
\end{figure}

The number of points located around $\Gamma$, when imposing the boundary condition, does not necessarily has to coincide with the number $\alpha$ of orthonormal functions, as it is shown in \cite{ckr}, where the collocation method is employed. Nevertheless, on behalf of simplicity, this work will consider $35$ equally distributed points around the perimeter of the unitary circle, and $\alpha=35$ base functions $u_{\alpha}$. The boundary condition will be obtained by evaluating the solution $u$, presented in (\ref{ns:02}), over this set of points.

The absolute error $\mathcal{E}$ is defined as the classical Lebesgue norm
\begin{equation}
\mathcal{E}=\left(\oint \left(u(l)-\sum_{\alpha=0}^{2N+1}b_{\alpha}u_{\alpha}(l) \right)^{2}dl\right)^{\frac{1}{2}};
\label{ns:03}
\end{equation}
where $u(l)$ represents the solution (\ref{ns:02}) valued on the boundary $\Gamma$, and $N=17$. The result of this integral will be approached using the standard trapezoidal method over $1000$ equally distributed points on the segment $[0,2\pi]$.

\begin{figure}
\centering
\subfigure[Absolute values of the 35 coefficients employed for approaching the boundary condition.]{
\includegraphics[scale=0.2]{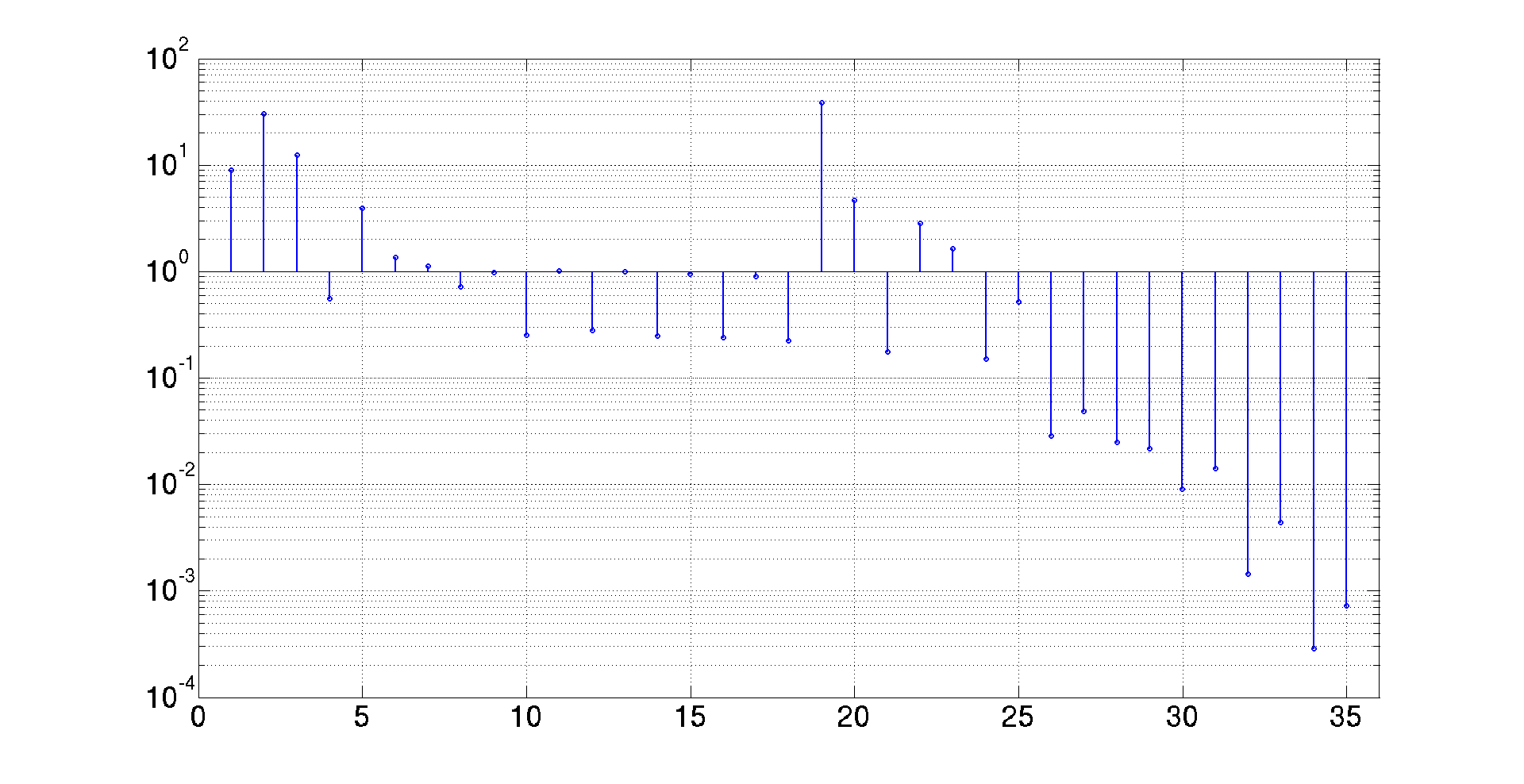}
\label{fig:ns01a}
}
\subfigure[Comparison between the boundary condition and the approached solution.]{
\includegraphics[scale=0.2]{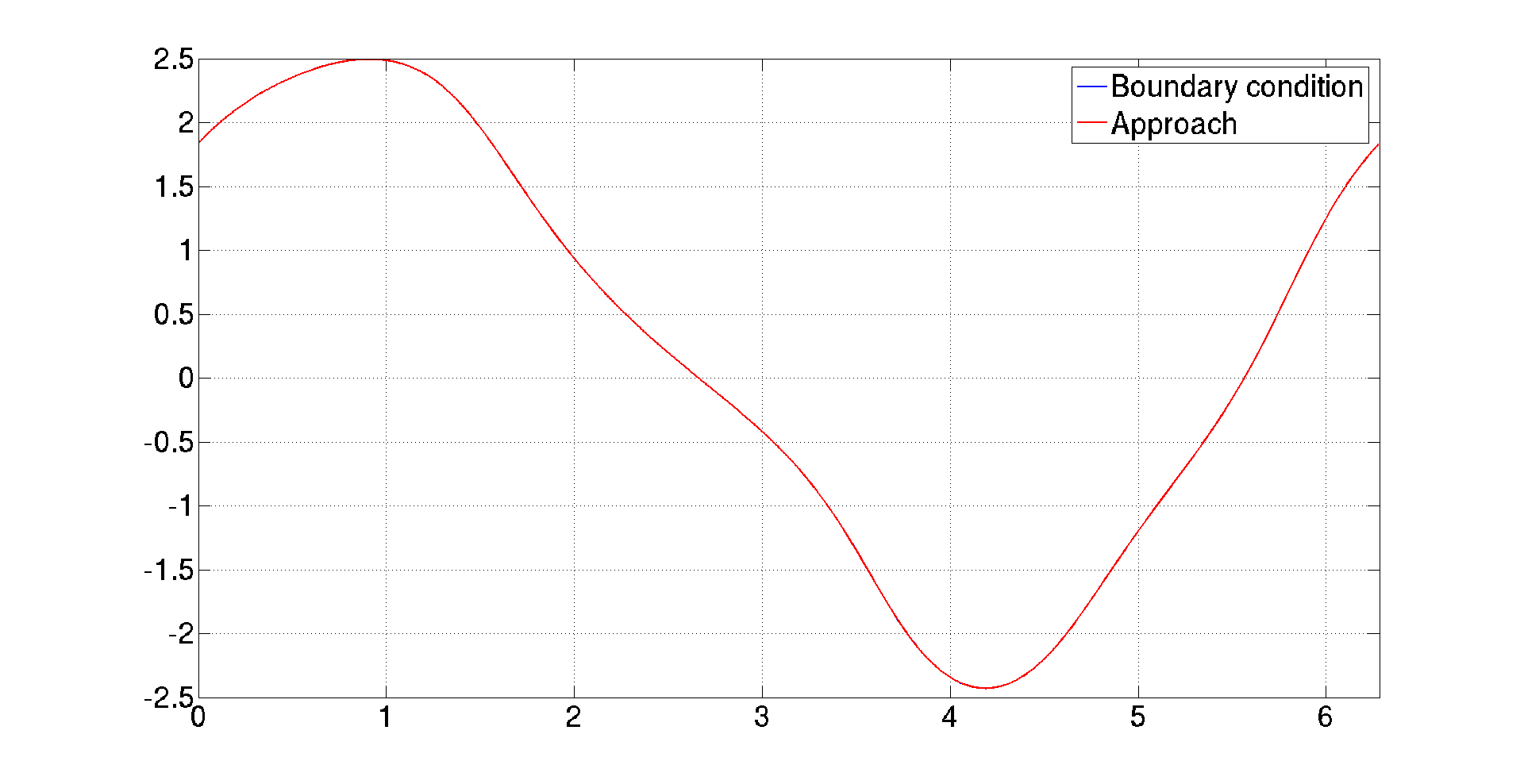}
\label{fig:ns01b}
}
\caption{Illustrations for the case $\sigma=(2+\cos\pi x)(2+\sin\pi y)$.}
\end{figure}

The Figure \ref{fig:ns01a} displays a logarithmic plot of the absolute magnitudes corresponding to every coefficient $b_{\alpha}$ used for approaching the boundary condition. We shall remember that the first $18$ coefficients correspond to orthonormal system obtained from the set (\ref{ns:04}), whereas the remaining $17$ correspond to the set (\ref{ns:05}). The Figure \ref{fig:ns01b} displays a comparison between the imposed boundary condition (\ref{ns:02}) valued on $\Gamma$, drew in blue, and the approached solution, in red. It is not possible to detect any difference at first sight in this case. Beside, the obtained error $\mathcal{E}=0.155\times 10^{-3}$, thus it is possible to assess the approach is adequate.

The Table \ref{tab:ns00} show the numerical values of some of the most relevant coefficients $b_{\alpha}$ used in the reconstruction. It is important to pay attention to the number $\alpha$ selected, since, in general, they are not presented in consecutive order.

\begin{table}
\caption{\label{tab:ns00}Values of the coefficients $b_{\alpha}$ corresponding to the boundary value problem with $\sigma=(2+\cos\pi x)(2+\sin\pi y)$.}
\centering
\begin{tabular}{| c | c | c | c | c | c | c | c |}\hline

$b_{0}$ & $b_{1}$ & $b_{2}$ & $b_{4}$ & $b_{19}$ & $b_{20}$ & $b_{22}$ & $b_{23}$\\ \hline

9.022 & 30.180 & 12.446 & 3.914 & -38.846 & 4.632 & 2.820 & 1.627 \\ \hline

\end{tabular}
\end{table}

\subsection{The Lorentzian cases.}

\begin{proposition}
\cite{ioprrh} Let the conductivity function be
\begin{equation}
\sigma=\left(\frac{1}{\left(x-\beta\right)^2+0.1}\right)\left(\frac{1}{ y^2+0.1}\right).
\label{ns:06}
\end{equation}
Thus the function
\begin{equation}
u=\frac{\left(x-\beta\right)^3+y^3}{3}+0.1\left(x-\beta +y\right),
\label{ns:07}
\end{equation}
is a solution of (\ref{int:00}).
\end{proposition}
Three cases will be considered: $\beta=0,\ 0.5,\ 1$. The Figures \ref{fig:ns02a}, \ref{fig:ns02b} and \ref{fig:ns02c} display each example.

\begin{figure}
\centering
\subfigure[Case when $\beta=0$.]{
\includegraphics[scale=0.25]{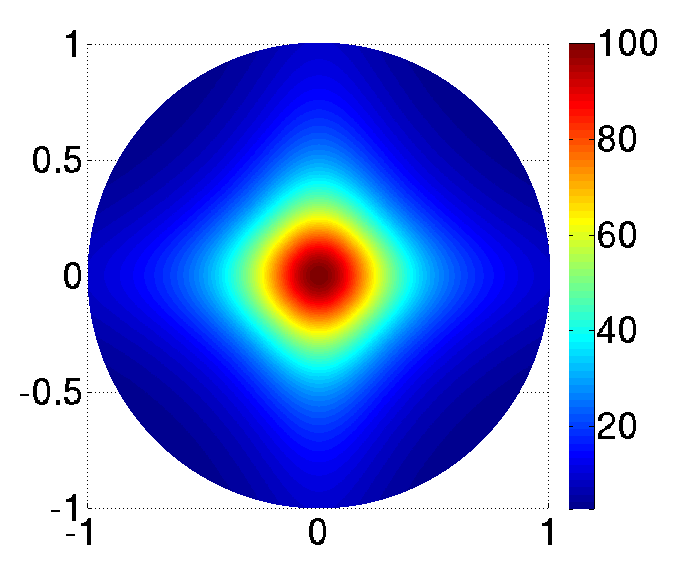}
\label{fig:ns02a}
}
\subfigure[Case when $\beta=0.5$.]{
\includegraphics[scale=0.25]{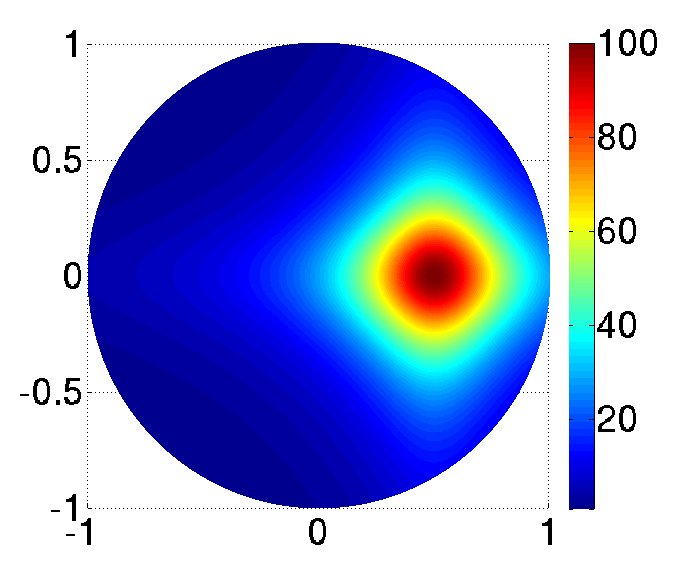}
\label{fig:ns02b}
}
\subfigure[Case when $\beta=1$.]{
\includegraphics[scale=0.25]{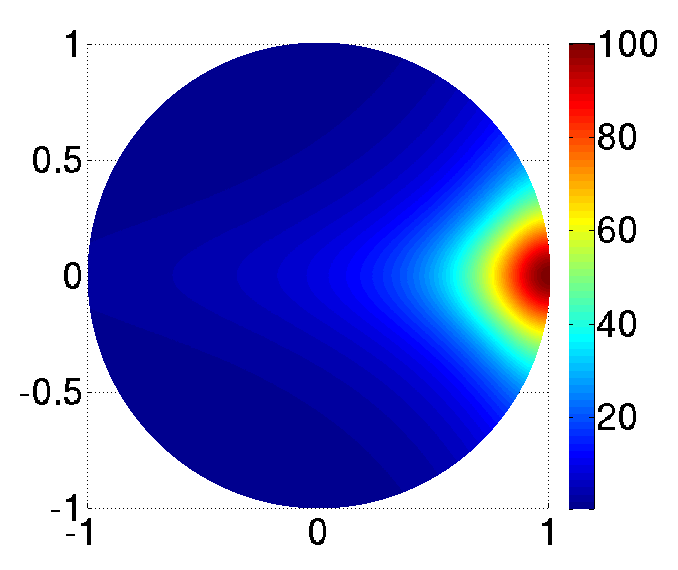}
\label{fig:ns02c}
}
\caption{Illustrations for the case $\sigma=((x-\beta)^{2}+0.1)^{-1}(y^{2}+0.1)^{-1}$.}
\end{figure}

Following the same logical steps of the previous subsection, the Figures \ref{fig:ns03a}, \ref{fig:ns03b} and \ref{fig:ns03c}, show logarithmic plots of the absolute values of the coefficients $b_{\alpha}$ corresponding to the cases $\beta=0,\ 0.5,\ 1$; respectively. The illustrations with the comparisons between the boundary conditions obtained from (\ref{ns:07}) and the approached solutions were omitted, because once more, it was not possible to notice any difference between the pairs of curves.

The absolute error for the case $\beta=0$ was $\mathcal{E}=0.592\times 10^{-3}$, the one corresponding to $\beta=0.5$ was $\mathcal{E}=1.4\times 10^{-3}$, and such belonging to $\beta=1$ was $\mathcal{E}=2.9\times 10^{-3}$. 

\begin{figure}
\centering
\subfigure[Absolute values of the 35 coefficients employed for approaching the boundary condition when $\beta=0$.]{
\includegraphics[scale=0.2]{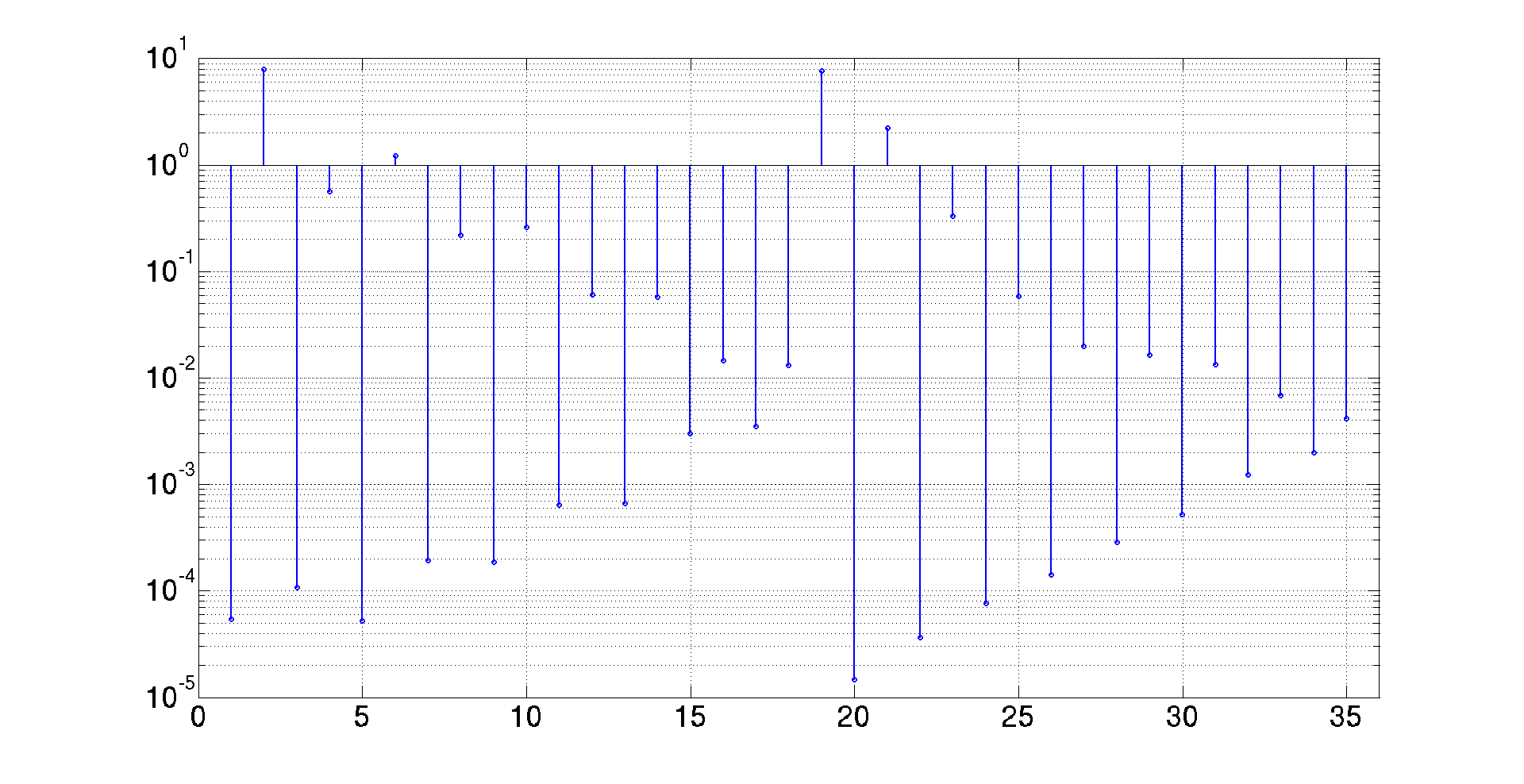}
\label{fig:ns03a}
}
\subfigure[Absolute values of the 35 coefficients employed for approaching the boundary condition when $\beta=0.5$.]{
\includegraphics[scale=0.2]{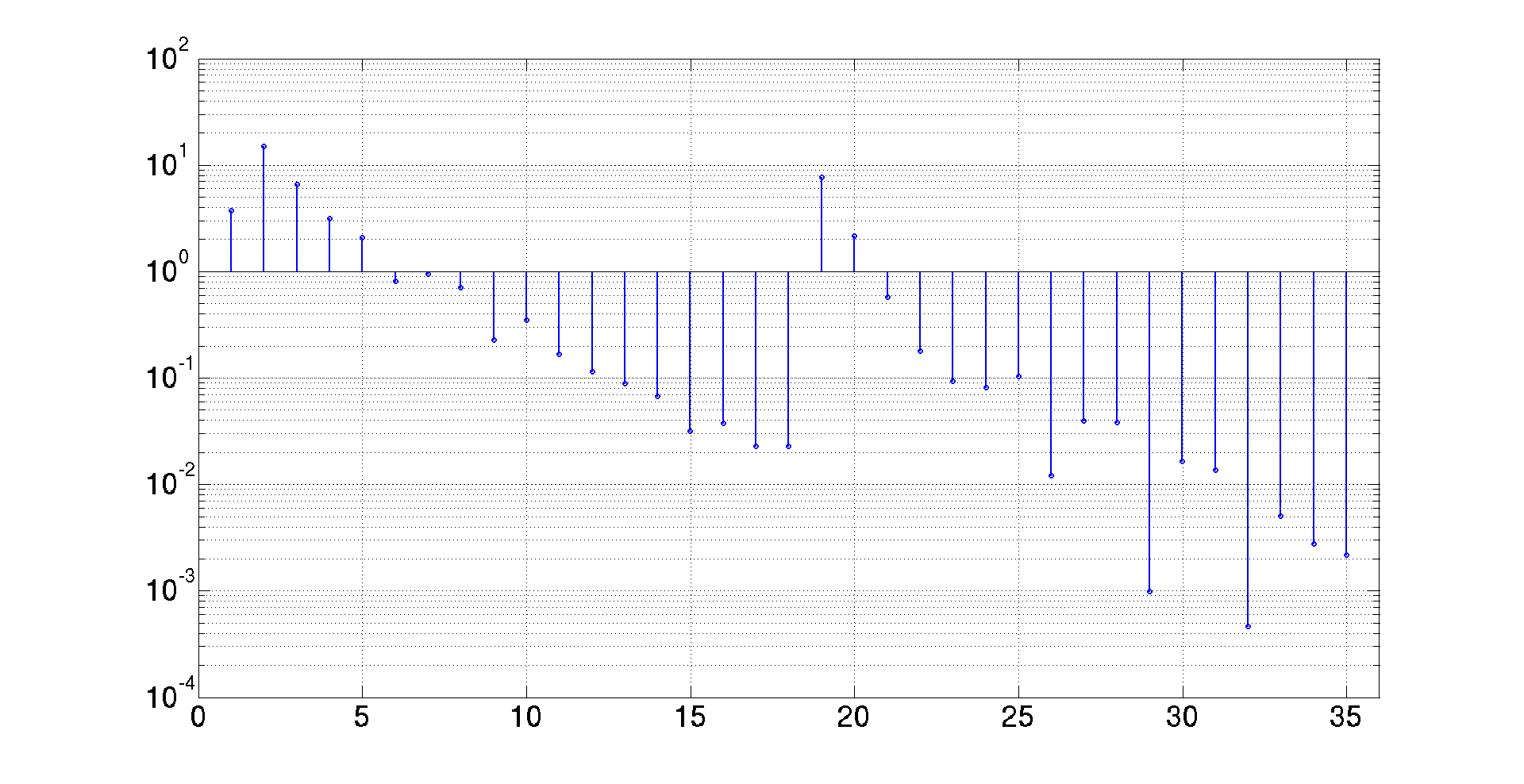}
\label{fig:ns03b}
}
\subfigure[Absolute values of the 35 coefficients employed for approaching the boundary condition when $\beta=1$.]{
\includegraphics[scale=0.2]{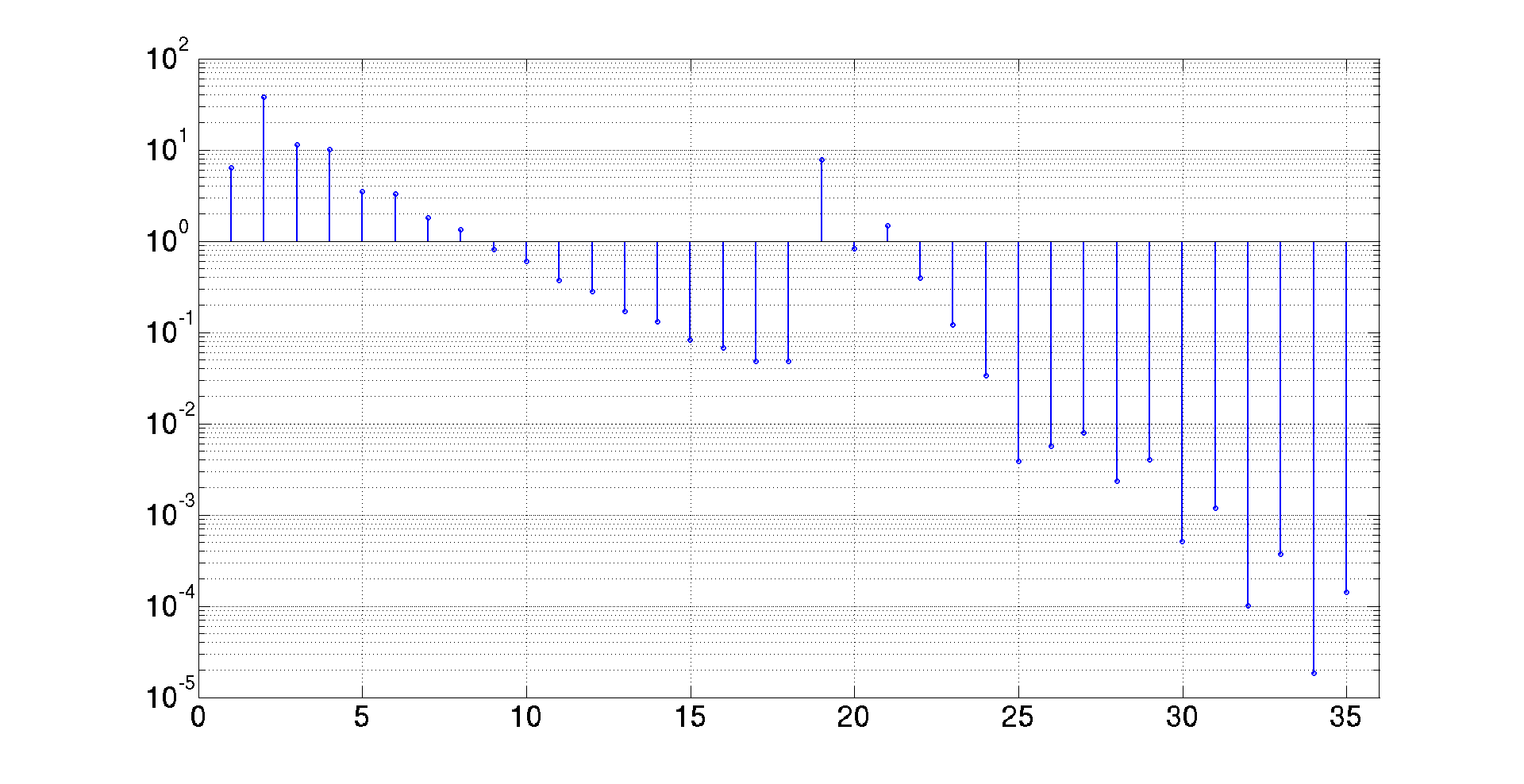}
\label{fig:ns03c}
}
\caption{Coefficients $b_{\alpha}$ for the cases $\sigma=((x-\beta)^{2}+0.1)^{-1}(y^{2}+0.1)^{-1}; \beta=1,\ 0.5,\ 1$.}
\end{figure}

The Tables \ref{tab:ns01}, \ref{tab:ns02} and \ref{tab:ns03} show some of the most relevant coefficients corresponding to the approached solutions of the cases $\beta=1,\ 0.5,\ 1$.

\begin{table}
\centering
\caption{\label{tab:ns01}Values of the coefficients $b_{\alpha}$ corresponding to the boundary value problem with $\sigma=(x^{2}+0.1)^{-1}(y^{2}+0.1)^{-1}$.}

\begin{tabular}{| c | c | c | c | c | c | c | c |}
\hline
$b_{1}$ & $b_{3}$ & $b_{5}$ & $b_{9}$ & $b_{19}$ & $b_{21}$ & $b_{23}$ & $b_{25}$\\
\hline
7.923 & 0.563 & -1.126 & -0.261 & -7.725 & 2.220 & -0.331 & 0.058 \\
\hline
\end{tabular}
\end{table}

\begin{table}
\centering
\caption{\label{tab:ns02}Values of the coefficients $b_{\alpha}$ corresponding to the boundary value problem with $\sigma=((x-0.5)^{2}+0.1)^{-1}(y^{2}+0.1)^{-1}$.}

\begin{tabular}{| c | c | c | c | c | c | c | c |}
\hline
$b_{0}$ & $b_{1}$ & $b_{2}$ & $b_{3}$ & $b_{19}$ & $b_{20}$ & $b_{21}$ & $b_{22}$\\
\hline
3.743 & 15.051 & 6.542 & -3.125 & -7.723 & 2.160 & 0.578 & -0.179 \\
\hline
\end{tabular}
\end{table}

\begin{table}
\centering
\caption{\label{tab:ns03}Values of the coefficients $b_{\alpha}$ corresponding to the boundary value problem with $\sigma=((x-1)^{2}+0.1)^{-1}(y^{2}+0.1)^{-1}$.}

\begin{tabular}{| c | c | c | c | c | c | c | c |}
\hline
$b_{0}$ & $b_{1}$ & $b_{2}$ & $b_{3}$ & $b_{4}$ & $b_{5}$ & $b_{6}$ & $b_{19}$\\
\hline
6.355 & 37.838 & 11.3959 & -10.118 & -3.517 & -3.311 & -1.815 & -7.856 \\
\hline
\end{tabular}
\end{table}

\subsection{Conductivities corresponding to geometrical distributions.}

Let us consider first a piecewise conductivity function, in polar coordinates, of the form:
\begin{displaymath}
   \sigma (x,y) = \left\{
     \begin{array}{lr}
       \ 100: &  r \in [0,0.2);\\
       \ 30\ : &  r \in [0.2,0.4);\\
       \ 20\ : &  r \in [0.4,0.6);\\
       \ 15\ : &  r \in [0.6,0.8);\\
       \ 30\ : &  r \in [0.8,1).
     \end{array}
   \right.   
\end{displaymath}
\begin{equation}
\label{ns:08}
\end{equation}
Here $r$ denotes the radio. The Figure \ref{fig:ns04} illustrates this conductivity. To select a boundary condition for the geometrical cases, without performing physical measurements, it is not a trivial task. Nevertheless, the conductivity defined in (\ref{ns:08}) shall be somehow related with the Lorentzian cases previously studied. For this reason, the boundary condition will be precisely the expression (\ref{ns:08}):
\[
u=\frac{x^3+y^3}{3}+0.1\left(x+y\right).
\]
\begin{figure}
\centering
\includegraphics[scale=0.30]{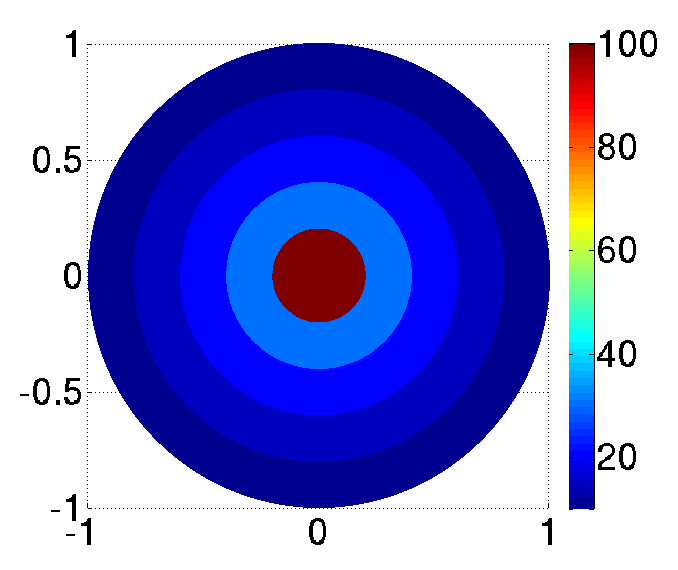}
\label{fig:ns04}
\caption{Conductivity defined according to the expression (\ref{ns:08}).}
\end{figure}

The Figure \ref{fig:ns05} displays the absolute value of the coefficients $b_{\alpha}$. Surprisingly, only four coefficients were significant for fulfilling the imposed boundary condition, and the total error was $\mathcal{E}=1.486\times 10^{-14}$. Again, the graphic comparing the boundary condition and the approach is omitted, since not any difference is visible.

\begin{figure}
\centering
\includegraphics[scale=0.20]{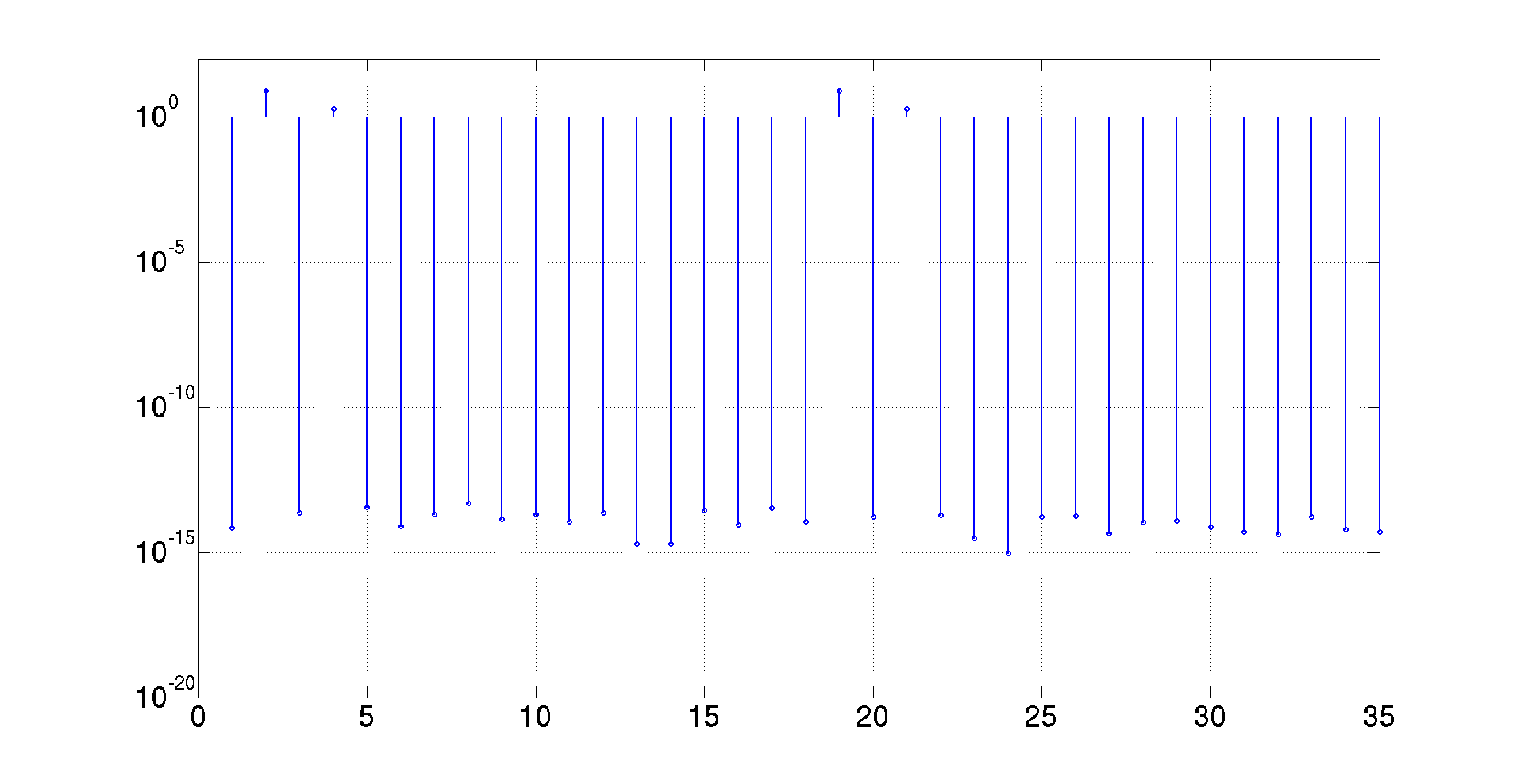}
\label{fig:ns05}
\caption{Conductivity defined according to the expression (\ref{ns:08}).}
\end{figure}

The Table \ref{tab:ns04} presents the values of the four most significant coefficients for this experiment, which also exhibit a very interesting symmetry.

\begin{table}
\centering
\caption{\label{tab:ns04}Values of the coefficients $b_{\alpha}$ corresponding to the boundary value problem with $\sigma$ defined in (\ref{ns:08}).}

\begin{tabular}{| c | c | c | c |}
\hline
$b_{1}$ & $b_{3}$ & $b_{19}$ & $b_{21}$ \\
\hline
7.826 & 1.863 & -7.826 & 1.863 \\
\hline
\end{tabular}
\end{table}

Let us consider now the conductivity functions illustrated in the Figures \ref{fig:ns06a}, \ref{fig:ns06b} and \ref{fig:ns06c}. The blue sections represent conductivity values of $\sigma=10$, whereas the red circles posses $\sigma=100$. The red disk in the Figure \ref{fig:ns06a}, corresponds to the equation
\[
x^2+y^2\leq 0.2,
\]
the red disk in the Figure \ref{fig:ns06b} is traced according to
\[
(x-0.6)^2+y^2\leq 0.2,
\]
whereas the red disk in the Figure \ref{fig:ns06c} is given by
\[
(x-0.79)^2+y^2\leq 0.2.
\]

\begin{figure}
\centering
\subfigure[Disk $x^2+y^2\leq 0.2$]{
\includegraphics[scale=0.25]{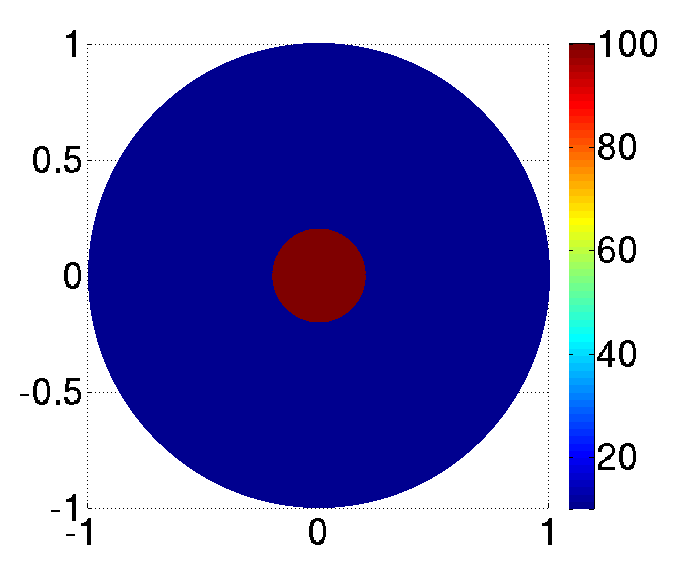}
\label{fig:ns06a}
}
\subfigure[Disk $(x-0.6)^2+y^2\leq 0.2$]{
\includegraphics[scale=0.25]{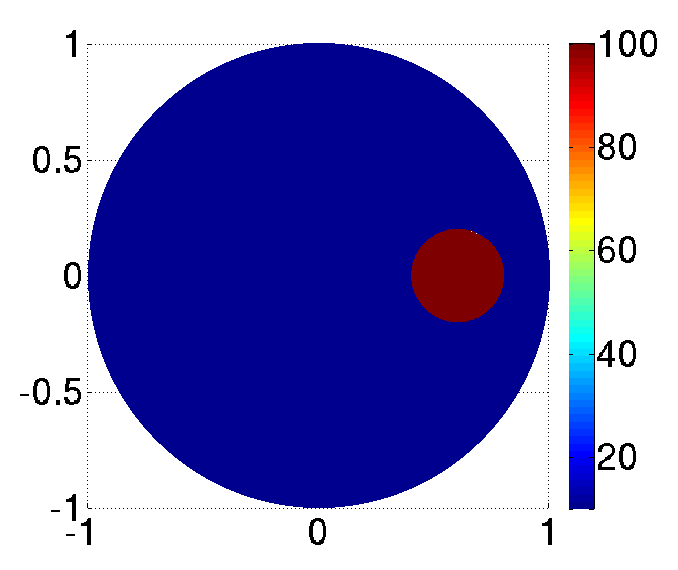}
\label{fig:ns06b}
}
\subfigure[Disk $(x-0.79)^2+y^2\leq 0.2$]{
\includegraphics[scale=0.25]{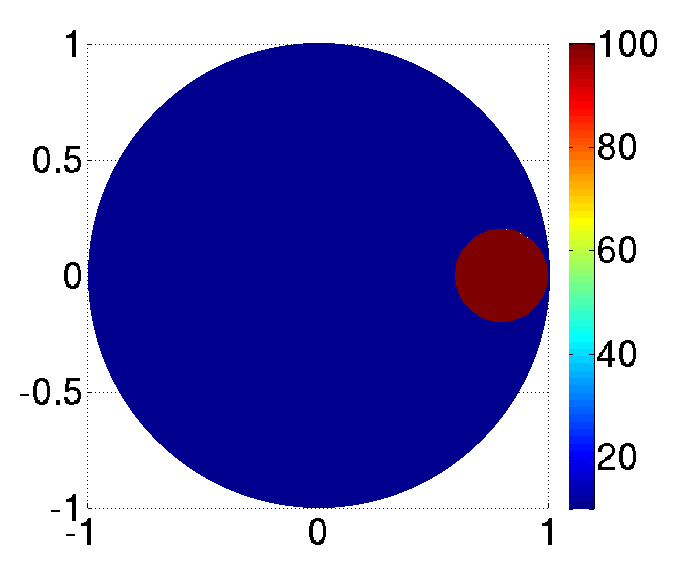}
\label{fig:ns06c}
}
\caption{Conductivity functions with a single disk within the domain. The red surfaces represent $\sigma=100$, whereas the blue sections denote $\sigma=10$.}
\end{figure}

The boundary condition for the conductivity illustrated in the Figure \ref{fig:ns06a} will be
\[
u=\frac{x^3+y^3}{3}+0.1\left(x+y\right),
\]
the condition for the conductivity function showed in the Figure \ref{fig:ns06b} will be
\[
u=\frac{(x-0.6)^3+y^3}{3}+0.1\left(x-0.6 +y\right);
\]
and finally, the condition imposed for the conductivity of the Figure \ref{fig:ns06c} was selected as
\[
u=\frac{(x-0.79)^3+y^3}{3}+0.1\left(x-0.79+y\right).
\]

The Figures \ref{fig:ns07a}, \ref{fig:ns07b} and \ref{fig:ns07c}, display the logarithmic graphics of the absolute values belonging to the coefficients $b_{\alpha}$. Once more, it is interesting that the quantity of significant values is relatively small. The absolute error approached for every case were: $\mathcal{E}=1.623 \times 10^{-14}$ for the case \ref{fig:ns06a}, $\mathcal{E}=7.3 \times 10^{-3}$ for the case \ref{fig:ns06b}, and $\mathcal{E}=3.9 \times 10^{-3}$ for the case \ref{fig:ns06c}.

\begin{figure}
\centering
\subfigure[Absolute values of the 35 coefficients employed for approaching the boundary condition of the conductivity \ref{fig:ns06a}.]{
\includegraphics[scale=0.2]{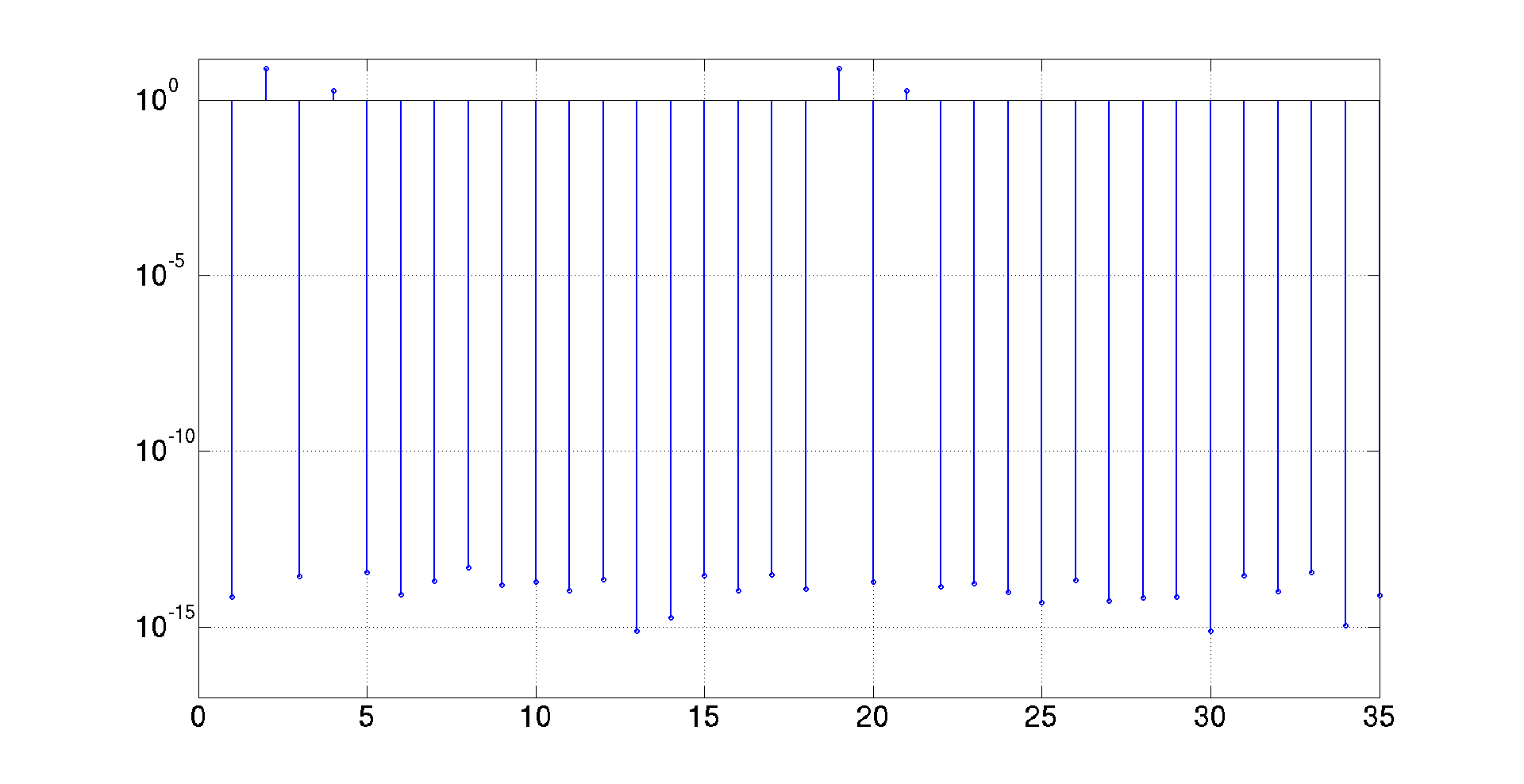}
\label{fig:ns07a}
}
\subfigure[Absolute values of the 35 coefficients employed for approaching the boundary condition of the conductivity \ref{fig:ns06b}.]{
\includegraphics[scale=0.2]{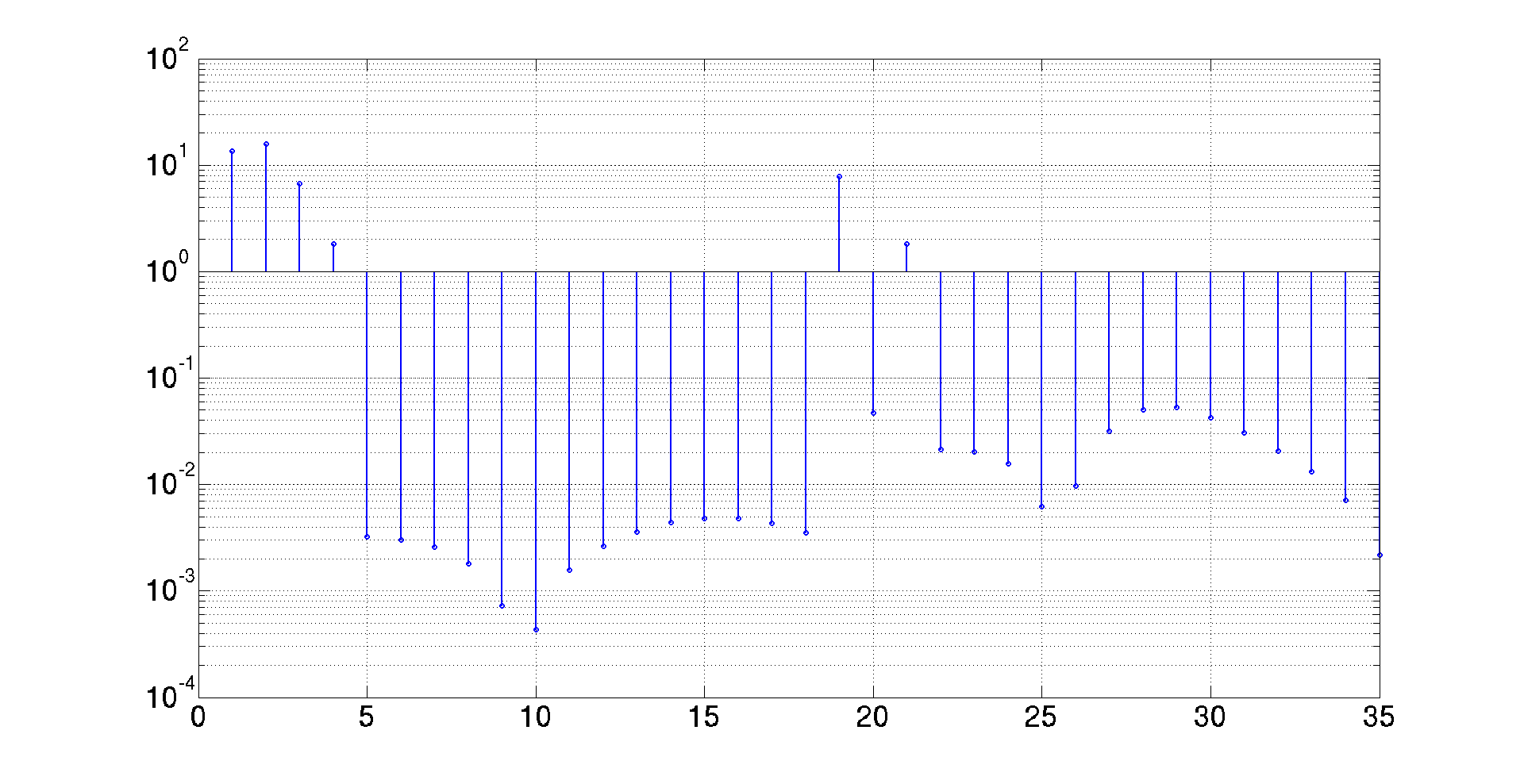}
\label{fig:ns07b}
}
\subfigure[Absolute values of the 35 coefficients employed for approaching the boundary condition of the conductivity \ref{fig:ns06c}.]{
\includegraphics[scale=0.2]{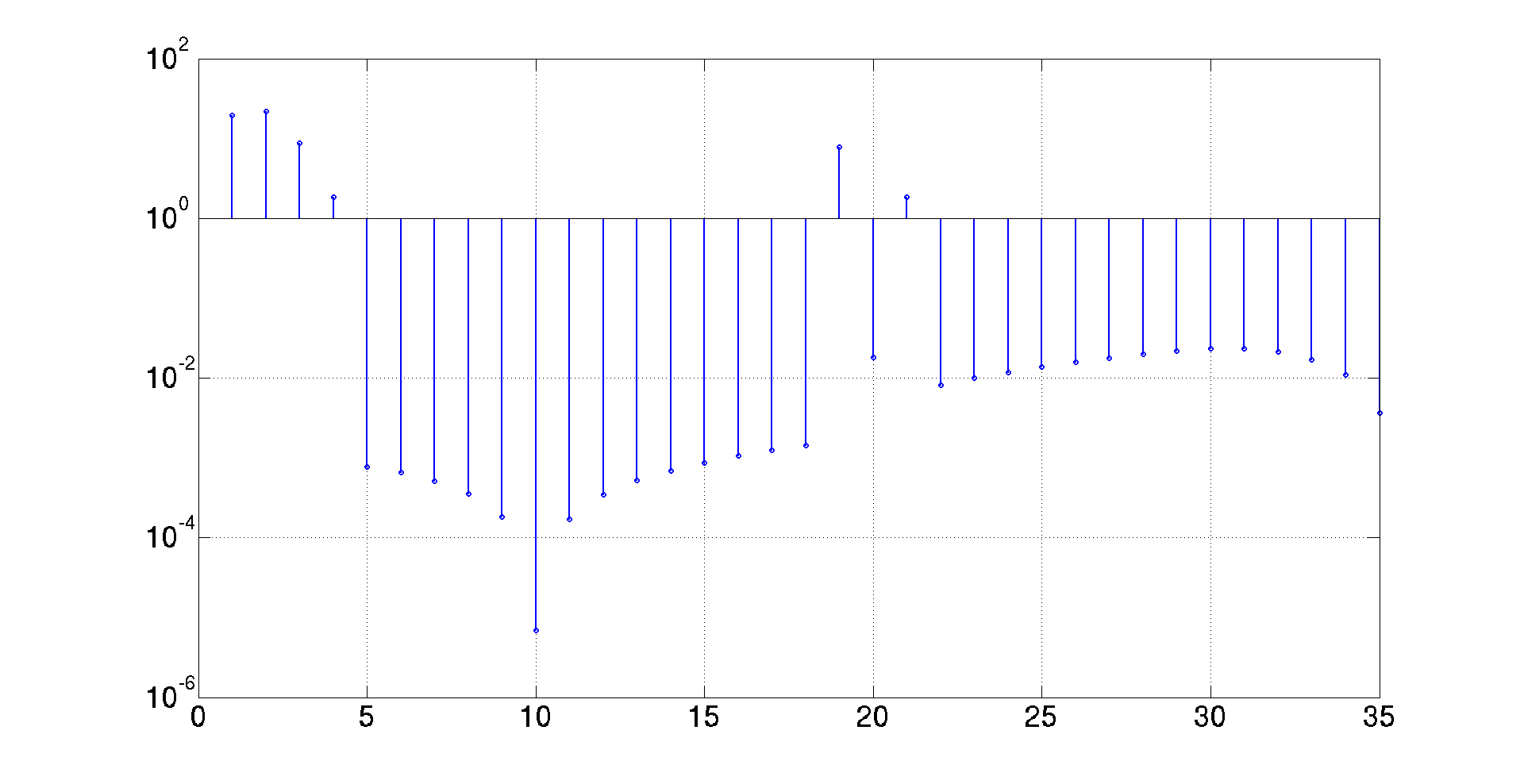}
\label{fig:ns07c}
}
\caption{Coefficients $b_{\alpha}$ for the cases of disks within the unitary circle.}
\end{figure}

As before, the Tables \ref{tab:ns05}, \ref{tab:ns06} and \ref{tab:ns07}, contain the values of the most significant coefficients for each of the last three examples.

\begin{table}
\centering
\caption{\label{tab:ns05}Values of the coefficients $b_{\alpha}$ corresponding to the boundary value problem with $\sigma$ posed in \ref{fig:ns06a}.}

\begin{tabular}{| c | c | c | c |}
\hline
$b_{1}$ & $b_{3}$ & $b_{19}$ & $b_{21}$ \\
\hline
7.826 & 1.863 & -7.826 & 1.863 \\
\hline
\end{tabular}
\end{table}

\begin{table}
\centering
\caption{\label{tab:ns06}Values of the coefficients $b_{\alpha}$ corresponding to the boundary value problem with $\sigma$ posed in \ref{fig:ns06b}.}

\begin{tabular}{| c | c | c | c | c | | c |}
\hline
$b_{0}$ & $b_{1}$ & $b_{2}$ & $b_{3}$ & $b_{19}$ & $b_{21}$ \\
\hline
-13.655 & 15.883 & -6.712 & 1.830 & -7.838 & 1.819 \\
\hline
\end{tabular}
\end{table}

\begin{table}
\centering
\caption{\label{tab:ns07}Values of the coefficients $b_{\alpha}$ corresponding to the boundary value problem with $\sigma$ posed in \ref{fig:ns06b}.}

\begin{tabular}{| c | c | c | c | c | c |}
\hline
$b_{0}$ & $b_{1}$ & $b_{2}$ & $b_{3}$ & $b_{19}$ & $b_{21}$ \\
\hline
-19.585 & 21.781 & -8.832 & 1.861 & -7.830 & 1.846 \\
\hline
\end{tabular}
\end{table}

\subsection{A triangular surface within the unitary circle.}

This is the last and one of the most interesting cases analysed in this work. We consider a triangle within the bounded domain, as displayed in the Figure 10, which somehow could be considered an interesting challenge when solving the Dirichlet problem of the two-dimensional Electrical Impedance Equation (\ref{int:00}), using numerical methods based upon variations of the Finite Element Method, one of the finest known numerical tools for this kind of problems.

\begin{figure}
\centering
\includegraphics[scale=0.30]{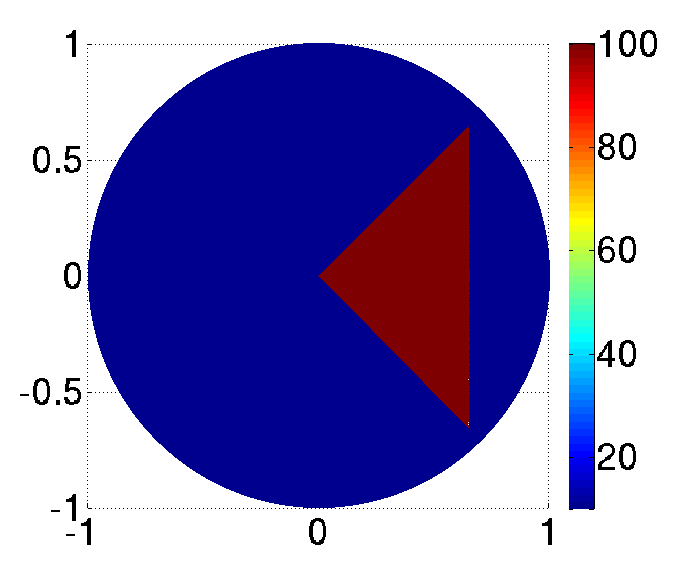}
\label{fig:ns08}
\caption{Conductivity containing a triangular figure. The red area represents $\sigma=100$, whereas the blue is $\sigma=10$.}
\end{figure}

The presence of corners requires special attention. Because of this, at least one radial trajectory in the calculations of the formal powers was force to cross over every corner of the triangle. Yet, the selection of the boundary condition is a completely different question. Indeed, this case will barely allow the imposition of theoretical conditions reaching acceptable convergence results.

For this reason, without denying the arbitrary selection of the function, the boundary condition was established as 
\[
u=\frac{(x-0.6)^3+y^3}{3}+0.1\left(x-0.6+y\right).
\]
This is indeed the only case for which $32$ formal powers were numerically approached, reaching a system of $61$ orthonormal functions. The total error was $\mathcal{E}=0.057$, but the graphic containing the comparison between the boundary condition and the approached solution, traced in Figure 11, deserves special attention.

\begin{figure}
\centering
\includegraphics[scale=0.2]{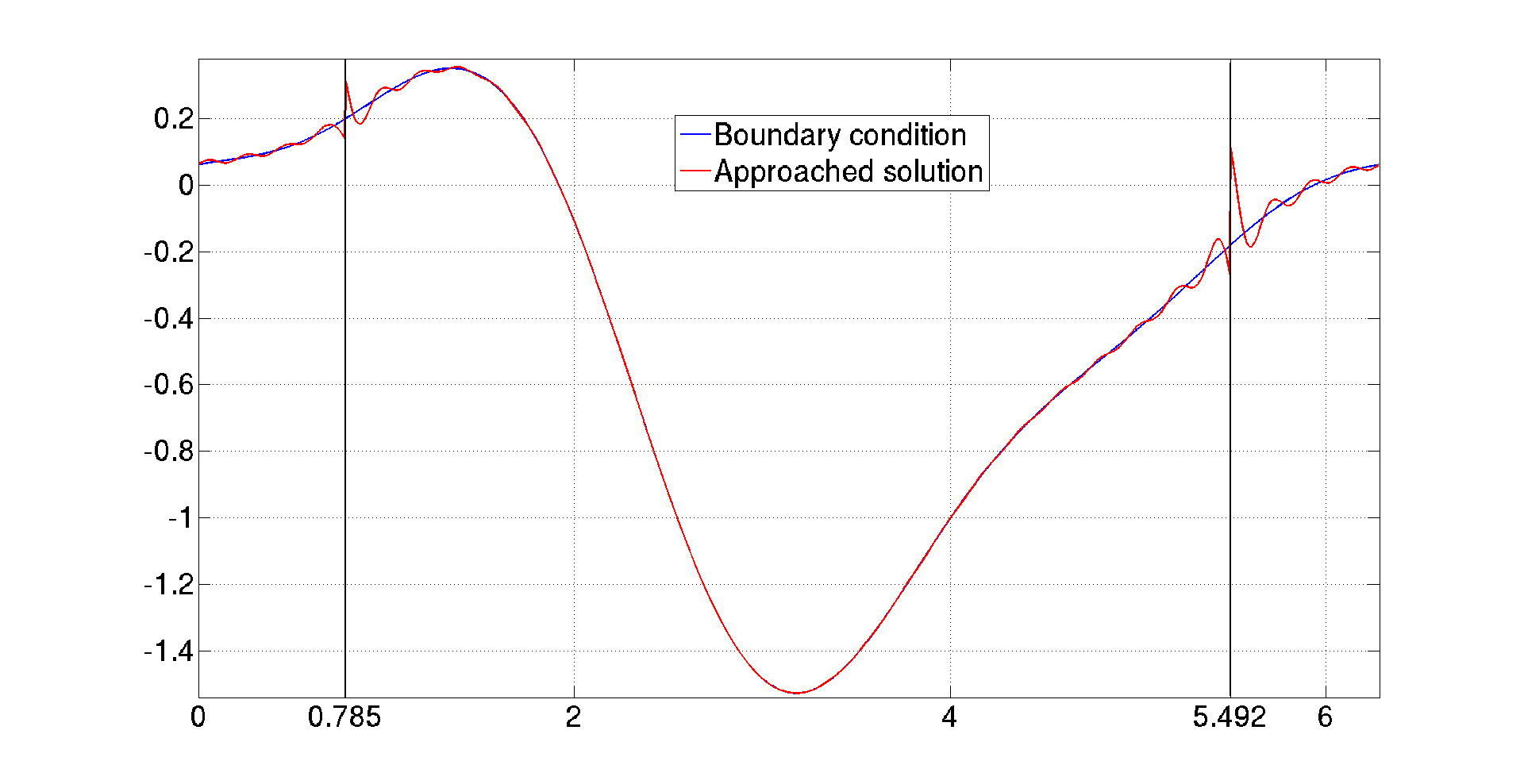}
\label{fig:ns09}
\caption{Comparison between the boundary condition and the approached solution of the conductivity function displayed in the Figure 10.}
\end{figure}

Let $l\in[0,2\pi)$. The highest divergence is located around the values $l=\frac{\pi}{4}$ and $l=\frac{5\pi}{4}$, as it is pointed out in the corresponding decimal representation on the Figure 11. These coincide precisely with the radii where the corners of the triangle are located, which posses interesting questions. Nevertheless, it is convenient to remember that the boundary condition was arbitrary selected, thus perhaps it is better to postpone the discussions until the moment that physical measurements provide the boundary conditions to be approached.

As for the rest of the posed examples, the Table \ref{tab:ns08} contains some of the most representative values of the coefficients $b_{\alpha}$ used for approaching the solution.

\begin{table}
\centering
\caption{\label{tab:ns08}Values of the coefficients $b_{\alpha}$ corresponding to the boundary value problem with $\sigma$ posed in Figure 10.}

\begin{tabular}{| c | c | c | c | c | c |}
\hline
$b_{0}$ & $b_{1}$ & $b_{2}$ & $b_{3}$ & $b_{32}$ & $b_{21}$ \\
\hline
-13.617 & 15.919 & -6.281 & 2.929 & -7.845 & 1.231 \\
\hline
\end{tabular}
\end{table}

\section{Conclusions}

On of the most important contributions of the present work is the full opening of the path for applying the Modern Theory of Pseudoanalytic Functions, into the analysis of wide class of conductivity functions upcoming from physical problems, by virtue of the Proposition \ref{pro:00}.

Tested in a variety of examples, the numerical analysis based upon the cited proposition, succeed to approach the imposed boundary conditions with considerable accuracy, presenting the highest divergence only in a kind of problem that is well known for its complexity when analysing boundary value problems for the two-dimensional Electrical Impedance Equation.

An immediate implication of these results is the possibility of analysing most classes of images corresponding to classical applications of the Electrical Impedance Tomography, as it is the Medical Imaging clinical monitoring. In this precise direction, it is possible to assure that the search for patterns of change in the boundary electric potentials, when changes in the conductivity within the domain of interest are taking place, is completely viable by applying the techniques posed before. It is important to remark that all presented results can be extended without mayor complications to a wide class of bounded domains, beside the unitary circle, that on behalf of simplicity was considered in these pages.

Indeed, most problems of Mathematical Physics, closely related with the Modern Pseudoanalytic Function Theory, as those elegantly studied in \cite{kpa}, as well as in \cite{hijar}, where a special case of the Fokker-Planck equation is posed, could well be susceptible for this kind of analysis.

Still, it is the Electrical Impedance Tomography problem one of the most interesting fields for applying the techniques suggested in this work.

\begin{acknowledgement} The author would like to acknowledge the support of CONACyT project 106722, and the support of HILMA S.A. de C.V., Mexico.
\end{acknowledgement}

\end{document}